\documentclass{article}
\usepackage{ijcai21}

\usepackage{times}
\usepackage{soul}
\usepackage{url}
\usepackage[hidelinks]{hyperref}
\usepackage[utf8]{inputenc}
\usepackage[small]{caption}
\usepackage{graphicx}
\usepackage{amsmath}
\usepackage{amsthm}
\usepackage{booktabs}
\usepackage[linesnumbered,algoruled,boxed,lined]{algorithm2e}

\newtheorem{theorem}{Theorem}

\title{On Computation Complexity of True Proof Number Search}

\author{
	Chao Gao
    \affiliations
    University of Alberta~\footnote{Extended work of a section in [Gao, 2020].}
    \emails
    cgao3@ualberta.ca
}

\begin{document}

\maketitle

\begin{abstract}
We point out that the computation of true \emph{proof} and \emph{disproof} numbers for proof number search in arbitrary directed acyclic graphs is NP-hard, an important theoretical result for proof number search. The proof requires a reduction from SAT, which demonstrates that finding true proof/disproof number for arbitrary DAG is at least as hard as deciding if arbitrary SAT instance is satisfiable, thus NP-hard.    
\end{abstract}

\section{Introduction}
Solving games is an important work direction for artificial intelligence research. Proof number search (PNS)~\cite{allis1994searching,allis1994proof} was developed specialized for solving games, drawing inspiration from conspiracy number search~\cite{mcallester1988conspiracy}. Unlike Alpha-Beta based fixed-depth pruning~\cite{knuth1975analysis}, PNS is a best-first search that iteratively adjusts its search path by \emph{proof} and \emph{disproof} numbers, making it a stronger alternative for solving games especially in the presence of deep and narrow plays. The notion of proof and disproof numbers for node $x$ is used to express the minimum number of descending leaf nodes $x$ has to solve in order to prove that $x$ is a win and loss respectively. 
In standard implementation, PNS initializes proof and disproof numbers of leaf nodes as $(1,1)$ and then backup these numbers recursively through an \emph{sum} operation, though there are enhancements trying to establish more informative initialization of proof and disproof numbers at node creation~\cite{breuker1998memory,breuker1999pn2,winands2004effective}.  
Depth-first proof number search (DFPN)~\cite{nagai2002df} is a PNS variant that adopts
two thresholds to avoid unnecessary traversal of the search tree; it has the same behavior as PNS in trees, but exhibits lower memory footprint at the expense of re-expansion; it can be further improved by incorporating various general or game-dependent techniques. Yoshizoe et al~\cite{yoshizoe2007lambda} introduced
$\lambda$ search to DFPN to solve the capturing problems in Go; 
threshold controlling and source node detection~\cite{kishimoto2010dealing} were introduced to DFPN to deal with a variety of issues from Tsume-Shogi. 

Together with other game-specific or game-independent algorithmic developments, PNS and their variants have been used for successfully solving a number of games, e.g.,  Gomoku~\cite{Allis1996GoMokuSB}, checkers~\cite{schaeffer2007checkers}, and small board size Hex~\cite{pawlewicz2013scalable,gao2017focused}. Furthermore, since PNS was developed by modeling game-searching as AND/OR trees~\cite{nilsson1980principles}, the algorithm has also been applied for solving real-world problems that can be described as a form of AND/OR graphs~\cite{pearl1984heuristics}. A notable example is     
chemical synthesis~\cite{heifets2012construction,kishimoto2019depth}. 

One problem in PNS is that the recursively computed proof and disproof numbers are well-defined on AND/OR trees~\cite{allis1994searching}, but in practice, the AND/OR structured state-space graph of many problems --- including many two-player games --- is  a directed acyclic graph (DAG). In some domains, it has been shown that treating these underlying graph as a tree may cause serious \emph{over-counting} for both proof and disproof numbers, resulting huge proof/disproof number for an easy-to-solve node, consequently preventing PNS from solving the domain for a long time~\cite{kishimoto2010dealing,nagai2002df}. Algorithms with exponential complexity are known to find the true proof and disproof numbers at each node in DAGs~\cite{schijf1994proof}, but they are impractical even for toy problems. Heuristic techniques address this by replacing sum-cost with a variant of max-cost at each node~\cite{ueda2008weak}, identifying some specific cases and curing them individually~\cite{kishimoto2010dealing,nagai2002df}. This over-counting problem has been discussed extensively in~\cite{kishimoto2012game,kishimoto2014recursive}; however, the theoretical complexity for dealing with this problem has not been presented by previous researchers. In this paper, we aim to formally establish the computational difficulty of true proof number search in DAGs; we show that such a task is NP-hard. Note that an earlier version of this discussion has appeared in a section in~\cite{cg2020Thesis}. This paper provides a specific account of the hardness of PNS as an independent work without the involvement of other unrelated topics.    

\section{Preliminaries}
To make this paper self-contained, in this section we review preliminaries for PNS, then discuss the over-counting issue of PNS in DAGs. 
\subsection{Directed Acyclic AND/OR Graphs}
A directed acyclic AND/OR graph~\cite{pearl1984heuristics} is a DAG with an additional property that any edge coming out of a node is  labeled either as an OR or AND edge.  A node contains only OR outgoing edges is called an \emph{OR node}. Conversely, a node emitting only AND edges is called an \emph{AND node}. Any node emanating both AND and OR edges is called \emph{mixed node}. To distinguish, in graphic notation, all AND edges from the same node are often grouped using an arc. It is also easy to see that a \emph{mixed node} can be replaced with two pure AND and OR nodes; see Figure~\ref{fig:mixed_node}. Thus, for ease of presentation, in the remaining text of this paper, we assume that a directed acyclic AND/OR graph contains only pure AND and OR nodes. That is, we note the graph as a tuple $G=<$$V_o, V_a, E$$>$, where $V_o$ and $V_a$ are respectively the set of OR and AND nodes, and $E$ represents the set of directed edges.     

Directed acyclic AND/OR graphs are often used to represent the problem-solving as a series of problem reduction. For example, Figure~\ref{fig:and_or_graph_example} can be interpreted as that ``to solve problem $A$, either $B$ and $C$ must be solved, to solve $B$, both $D$ and $E$ have to be solved, while for solving $C$, only $E$ or $F$ needs to be solved''.

\begin{figure}[tph]
\centering
\includegraphics[scale=0.8]{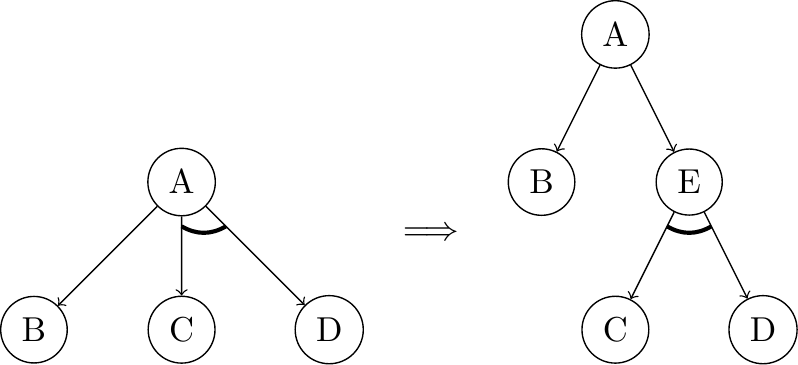}
\caption{Left: $A$ is a mixed node. Right: $A$ is an OR node; $E$ is an AND node. These two graphs are equivalent while the right one contains only pure AND and OR nodes.}
\label{fig:mixed_node}
\end{figure}

\begin{figure}[tph]
\centering
\includegraphics[scale=0.8]{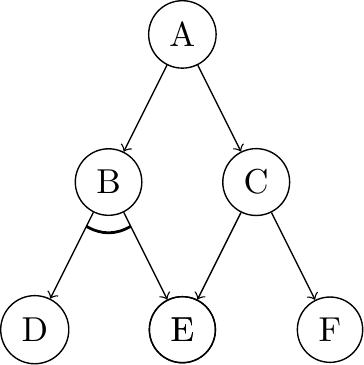}
\caption{Directed acyclic AND/OR graph represents problem-reduction. $B$ is the only AND node. $A$ and $C$ are OR nodes. $D$, $E$, and $F$ are leaf nodes that can be regarded as either AND or OR nodes.}
\label{fig:and_or_graph_example}
\end{figure}

Figure~\ref{fig:and_or_graph_example} also shows that if we assume $A$ is \emph{solvable}, in the best case, only one sub-problem $E$ is required to be \emph{solvable} to validate our assumption. Conversely, knowing only $E$ and $F$ both are \emph{unsolvable} is enough to say that $A$ is \emph{unsolvable}. In other words, the minimum number of leaf nodes to examine for \emph{proving} $A$ is 1, and $2$ for \emph{disproving}. The sub-graph that used to claim either $A$ is solvable or unsolvable is called \emph{solution-graph}. For Figure~\ref{fig:and_or_graph_example}, a \emph{solvable} solution-graph can be $\{A \to C \to E\}$, assuming $E$ is solvable; an \emph{unsolvable} solution-graph can be $\{A \to C \to E, C \to F\}$ assuming both $E$ and $F$ are unsolvable.

\subsection{Definition of Proof and Disproof Numbers}

Formally, given a graph $G$, $\forall x \in \{V_o, V_a\}$, define $p(x)$ and $d(x)$ respectively as the minimum number descending leaf nodes in order to solve for proving and disproving $x$ respectively.  
When graph $G$ is an AND/OR tree, the following recursive computation scheme exists~\cite{allis1994searching}:

\begin{equation} 
\begin{array}{l}
p(x) = 
\begin{cases}
1 \qquad \mbox{$n$ is non-terminal leaf node} \\ 
\min_{x_j \in ch(x)} p(x_j) \quad \mbox{$x$ is OR node} \\ 
\sum_{x_j \in ch(x)} p(x_j) \quad \mbox{$x$ is AND node} \\
\end{cases}
\\ \\ 
d(n)=
\begin{cases}
1 \quad \mbox{$x$ is non-terminal leaf node} \\
\min_{x_j \in ch(x)} d(x_j) \quad \mbox{$x$ is AND node} \\ 
\sum_{x_j \in ch(x)} d(x_j) \quad \mbox{$x$ is OR node} \\
\end{cases}
\end{array}
\label{eq:pns1}
\end{equation}
In Eq.~\eqref{eq:pns1}, $ch(x)$ represents the set of direct child successors for $x$. One scenario Eq.~\eqref{eq:pns1} fails to cover is when $x$ is a terminal leaf node, in which case proof and disproof numbers of $x$ are self-evident, as in Eq.~\eqref{eq:pns2}.  
\begin{equation}
\begin{array}{l}
p(x) = 
\begin{cases}
0 \qquad \mbox{$x$ is solvable} \\
\infty \qquad \mbox{$x$ is unsolvable} \\ 
\end{cases} 
\\ \\
d(x)=
\begin{cases}
0 \qquad \mbox{$x$ is unsolvable} \\
\infty \qquad \mbox{$x$ is solvable} \\ 
\end{cases}
\end{array}
\label{eq:pns2}
\end{equation}

Equipped with proof and disproof numbers, PNS conducts a best-first search repeatedly doing the following steps: 
\begin{enumerate}
\item Selection. Starting from the root, at each node $x$: 1) if $x$ is OR node, select a child node with the minimum $p$ value; 2) if $x$ is AND node, select a child node with the minimum $d$ value. Stop this until until $x$ becomes a leaf node.  
\item Evaluation and Expansion. An external function is called to check if the leaf is a terminal or not. If not, the leaf node is expanded and all its newly children are assigned with $(p, d) \gets (1,1)$.   
\item Backup. Updated proof and disproof numbers for the selected leaf node is back-propagated up to the tree according to~Eq.~\eqref{eq:pns1} and~\eqref{eq:pns2}.
\end{enumerate}

Given sufficient memory and computation time, it has been shown that PNS and DFPN are not only \emph{complete} for AND/OR trees but also acyclic AND/OR graphs~\cite{kishimoto2008completeness,allis1994searching}. That is, PNS terminates when the root node becomes a \emph{terminal}, i.e., its $(p,d)$ becomes either $(0,\infty)$ or $(\infty, 0)$, indicating respectively the root is \emph{solvable} or \emph{unsolvable}.

\subsection{Game-Tree and Game-DAG}
An AND/OR tree can be used to model game-tree~\cite{nilsson1980principles}, in which case, the state nodes where it is the first-player to play are OR nodes, and those of the second-player are AND nodes. 
Compared to general AND/OR trees, the additional regularity for an AND/OR tree of a two-player alternate-turn zero-sum game is that OR and AND  appear alternately in layers. That is, if $x$ is OR node, $\forall y \in ch(x)$, $y$ must be an AND node, and vice versa.

In this game context, the notion of $p(x)$ and $d(x)$ in Eq.~\eqref{eq:pns1} and~\eqref{eq:pns2} can be simplified as $\phi(x)$ and $\delta(x)$, which respectively represent the minimum number of non-terminal \emph{leaf} nodes to \emph{solve} in order to prove that $x$ is winning and losing.  Here, a node is said to be a \emph{winning} state if the player to play at that node wins (respectively for \emph{losing}).
By such, it becomes unnecessary to explicitly discern whether $x$ is an AND or OR node, since 
$\phi(x)$ would be solely dependent on $delta$ values of $ch(x)$, and $\delta(x)$ would be equal to the summed $\phi$ values of $ch(x)$.  
This simplified computation scheme is precisely expressed in Eq.~\eqref{eq:pns_phi_delta}. 
\begin{equation}
\begin{array}{l}
\phi(x) = 
\begin{cases}
1 \qquad \mbox{$x$ is non-terminal leaf node} \\ 
0 \qquad \mbox{$x$ is terminal winning state} \\ 
\infty \qquad \mbox{$x$ is terminal losing state} \\ 
\min\limits_{x_j \in ch(x)}  \delta(x_j) \\ 
\end{cases}
\\ \\
\delta(x)=
\begin{cases}
1 \qquad \qquad \mbox{$x$ is non-terminal leaf node} \\
\infty \qquad \mbox{$x$ is terminal winning state} \\ 
0 \qquad \mbox{$x$ is terminal losing state} \\ 
\sum_{x_j \in ch(x)} \phi(x_j)  \\
\end{cases}
\end{array} \label{eq:pns_phi_delta}
\end{equation}

Figure~\ref{fig:pns_example_phi_delta} shows an example \emph{game-tree}, where, according to PNS, node $j$ is to be selected for next node expansion, and after that, the ancestor nodes of $j$ will be updated due to the change in $j$. Note that, for better convenience, in games context, nodes of first players and second-player are respectively drawn using \emph{square} and \emph{circular} shapes, 
eliminating the use of arcs between edges for representing AND nodes.

\begin{figure}[tph]
\centering
\includegraphics[scale=0.8]{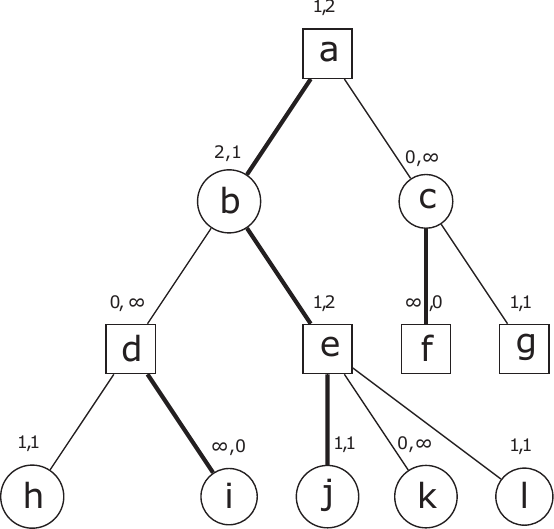}
\caption{PNS example in a game-tree. Each node has a pair of evaluations $(\phi, \delta)$, computed bottom up. A bold edge indicates a link where the minimum selection is made at each node, according to children's $\delta$ values.}
\label{fig:pns_example_phi_delta}
\end{figure}

In games, the existence of transpositions indicates that the state-space graphs are often a DAG rather than a tree. Although the computational convenience as in Eq.~\eqref{eq:pns_phi_delta} still holds due to the alternating regularity, as Eq.~\eqref{eq:pns1} for general AND/OR graphs, $\phi$ and $\delta$ no long represent true proof and disproof numbers. In the next section, we discuss the issue in detail.  

\subsection{Over-Counting in DAGs}
In a game-DAG, computing proof and disproof numbers via Eq.~\eqref{eq:pns_phi_delta} could over-count some non-terminal \emph{leaf} nodes multiple times. Figure~\ref{fig:pns_example_dag} shows an example where $E$ was counted twice at $A$. 

\begin{figure}[tp]
\centering
\includegraphics[scale=0.8]{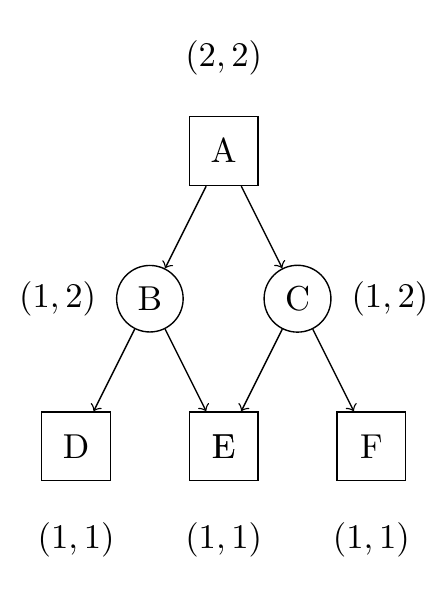}
\caption{PNS example in a game DAG. Solving $E$ to be a loss would imply both $B$ and $C$ are \emph{winning}, thus $A$ would be a loss; this implies that the $\delta$ value of $A$ should be $1$, but computation according~Eq.~\eqref{eq:pns_phi_delta} gives $\delta(A)=2$. This is because when summing the $\phi$ from $B$ and $C$, $E$ was counted twice.}
\label{fig:pns_example_dag}
\end{figure}

The over-counting problem in general AND/OR graphs can be extremely severe, since it is possible that a node might be counted an exponential number of times. 
Figures~\ref{fig:tree_lattice} and~\ref{fig:comb_lattice} are two examples where we can analytically see the drastic difference between the true proof number and the recursively computed number by Eq.~\eqref{eq:pns1}.     

\begin{figure}[tp]
\centering
\includegraphics[width=0.5\textwidth]{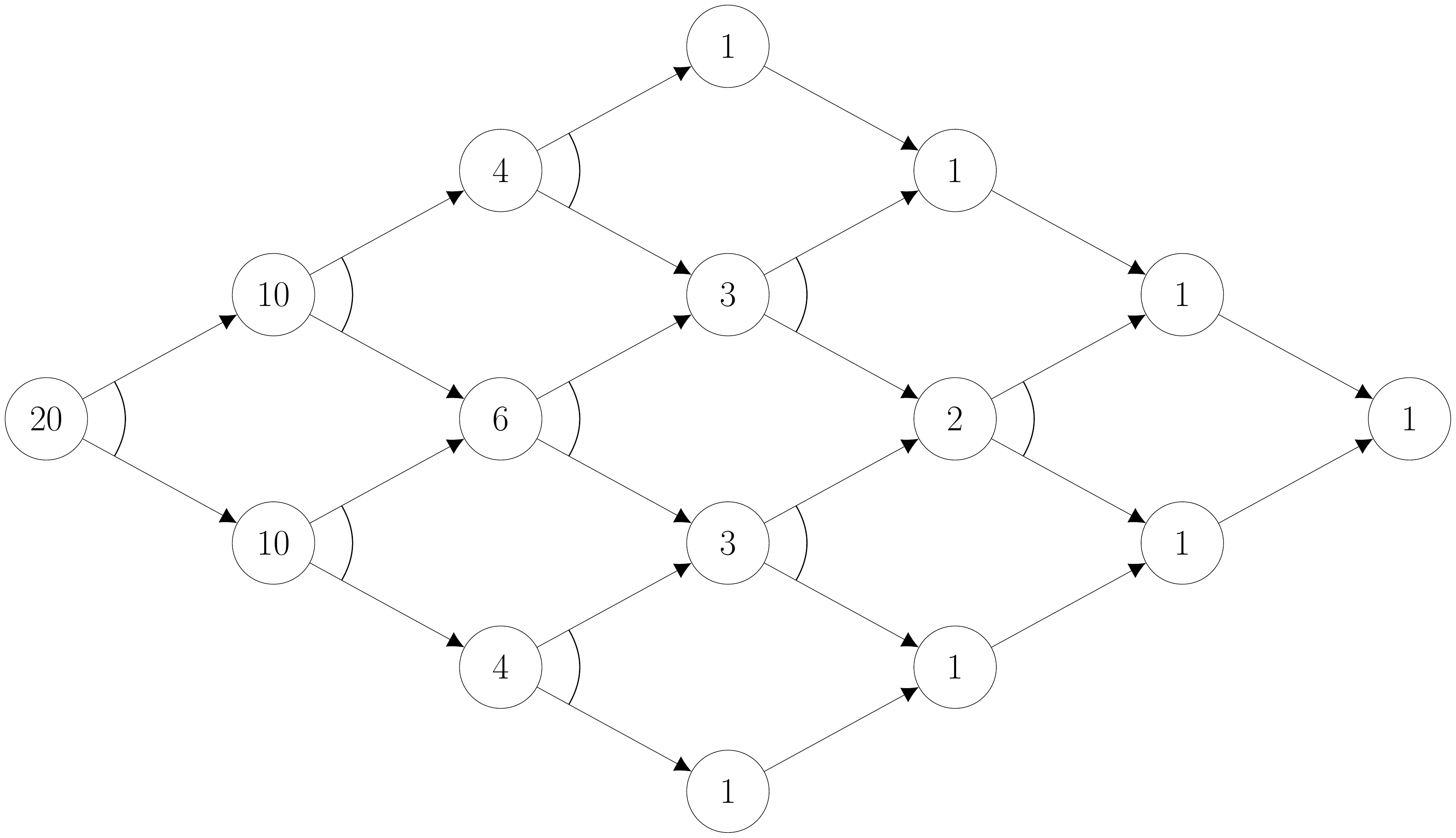}
\caption{An AND/OR graph in lattice form of $7$ layers. Every node is AND node.
For a graph of such having $n$ layers, the true proof number for the root is $1$, but computing it with Eq.~\eqref{eq:pns1} gives $\binom{n-1}{k}$, where $k=\frac{n-1}{2}$.} 
\label{fig:tree_lattice}
\end{figure}

\begin{figure}[tp]
\centering
\includegraphics[width=0.5\textwidth]{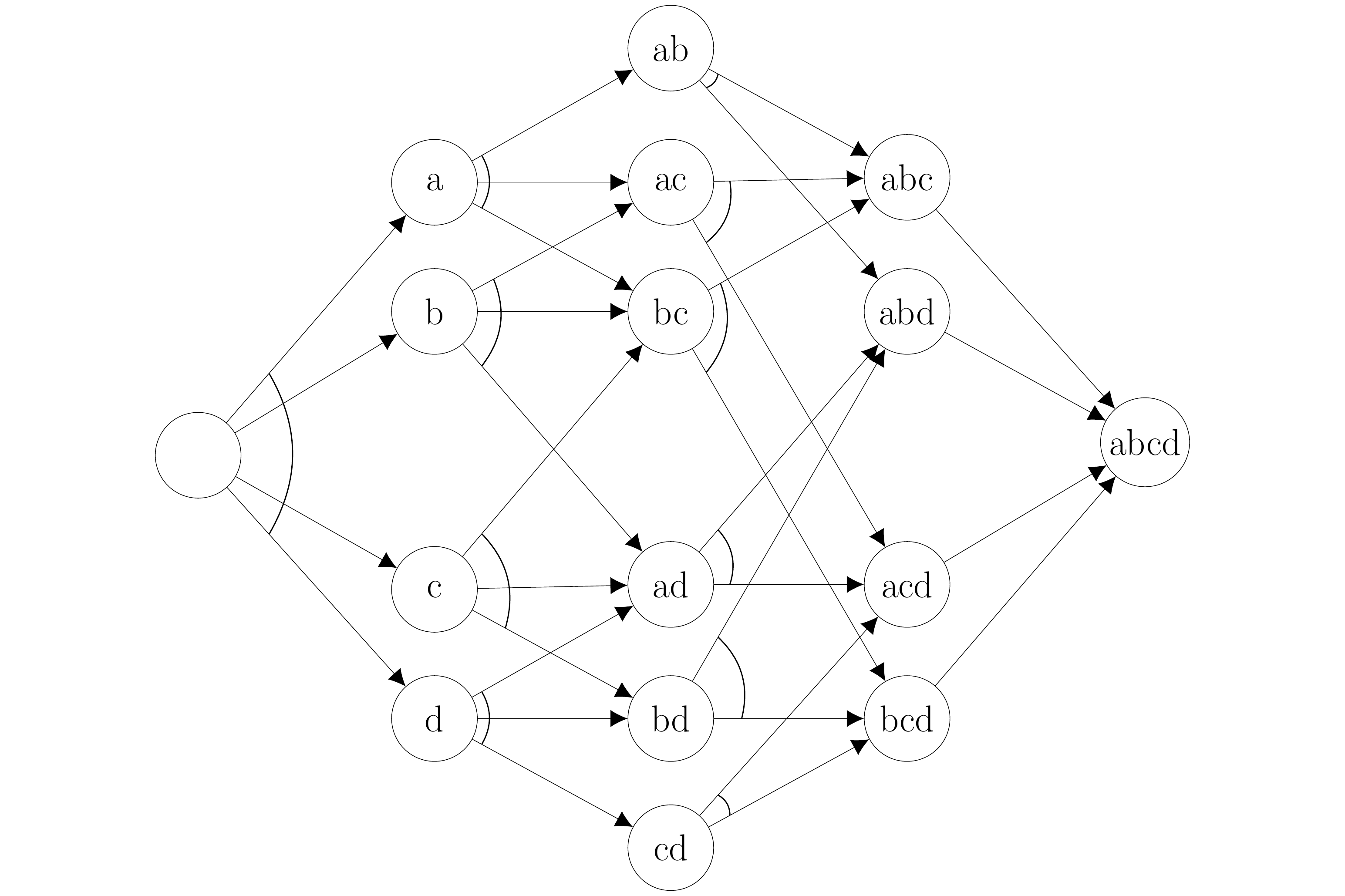}
\caption{An AND/OR graph in combinatorial lattice form of 5 layers.
For a graph of such having $n$ layers, the true proof number for the root is $1$, but computing it with Eq.~\eqref{eq:pns1} gives $(n-1)!$. } 
\label{fig:comb_lattice}
\end{figure}

A natural question then arises: for a layered DAG rooted at node $x$, what is the computational difficulty of calculating the true proof and disproof numbers for $x$? Even though many exact~\cite{schijf1994proof,muller2002proof} or heuristic~\cite{ueda2008weak,kishimoto2010dealing} methods have been carried out to address the \emph{over-counting} issue, formal proof on the theoretical difficulty of this task has been lacking. In the next section we prove this sub-task in proof number search is NP-hard. 

\section{NP-Hardness of Exact Computation}
We now show that the computation of true proof and disproof numbers in arbitrary AND/OR DAGs. This proof is closely related to proof of computationally difficulty for ``optimal solution graph'' from~\cite{sahni1974computationally}. Our contribution is bringing the classic results into the context of computing proof and disproof numbers, providing the heuristic search community a definite answer to an important question concerning PNS. 

First, we have the following observation.
\begin{theorem}
	Deciding whether a SAT instance is satisfiable can be reduced to finding the true proof (or disproof) number of an AND/OR graph.
	\label{theorem:sat_to_and_or_graph}
\end{theorem}
\begin{proof}
	Consider a SAT instance in \emph{conjunction normal form} $$P=\wedge_{i}^{k} C_i,$$ where each clause is a disjunction of literals, $C_i=\vee_{j} l_j$. Let $x_1, x_2, \ldots, x_n$ be all the variables, each literal $l_j$ is either $x_j$ or $\neg x_j$. Construct an AND/OR graph as follows. 
	\begin{enumerate}
		\item Let the start node be an AND node, denoted by $P$.
		\item $P$ contains $n+k$ successors $C_i$ and $X_j$, $\forall i=1, \ldots, k,~\forall j=1, \ldots, n$. They are all OR nodes.
		\item Each $X_j$ contains two successors $Tx_j$ and $Fx_j$ representing $x_j$ and $\neg x_j$ respectively. The successors of each $C_i$ are those literals that appear in that clause. 
	\end{enumerate}
	For such a graph, to satisfy the start node $P$, every clause must be satisfied and each variable node $X_j$ must be assigned to a value; thus, the minimum possible proof number for $P$ is $n$ --- in such a best case, to satisfy each $X_j$, only one of $Tx_j$ and $Fx_j$ is needed. This is equivalently to say finding if a SAT instance is satisfiable can be transformed into finding the true proof number of $P$ in this AND/OR graph.

	For disproof number, construct another graph by reverting the above graph such that all AND nodes are converted to OR nodes, and all OR nodes to AND nodes, then we see that to disprove $P$, it is sufficient to disprove either one in $\{C_1, \ldots, C_k, X_1, \ldots, X_n\}$; however, the minimum possible solution could be disprove $C_i$ if $\exists i =1, \ldots, k$ such that $C_i$ contains less than $2$ literals, otherwise $X_j$, $\forall j =1,\ldots, n$. For either case, this would lead the SAT instance unsatisfiable. So, this is equivalent to say finding if a SAT instance is unsatisfiable can be converted into finding the true disproof number of $P$ in the constructed AND/OR graph.  
\end{proof}

An example formula $P=(x_1 \vee x_2 \vee x_3 ) \wedge (\neg x_1 \vee \neg x_2 \vee \neg x_3) \wedge (\neg x_1 \vee x_2)$ and the constructed AND/OR graphs are shown in Figure~\ref{fig:sat_to_andor_example}. 

Then, we can derive the following result.
\begin{theorem}
Given arbitrary AND/OR DAG rooted at $x$, computing the true proof and disproof number for $x$ is NP-hard.
\end{theorem} \label{theorem:pns_np_hard}
\begin{proof}
The graph construction from SAT in Theorem~\eqref{theorem:sat_to_and_or_graph} is with polynomial time. It follows that computing proof or disproof number exactly in an arbitrary DAG is at least as difficult as finding if an SAT is satisfiable, thus NP-hard.
\end{proof}

\begin{figure}[tp]
	\centering
	\includegraphics[scale=0.5]{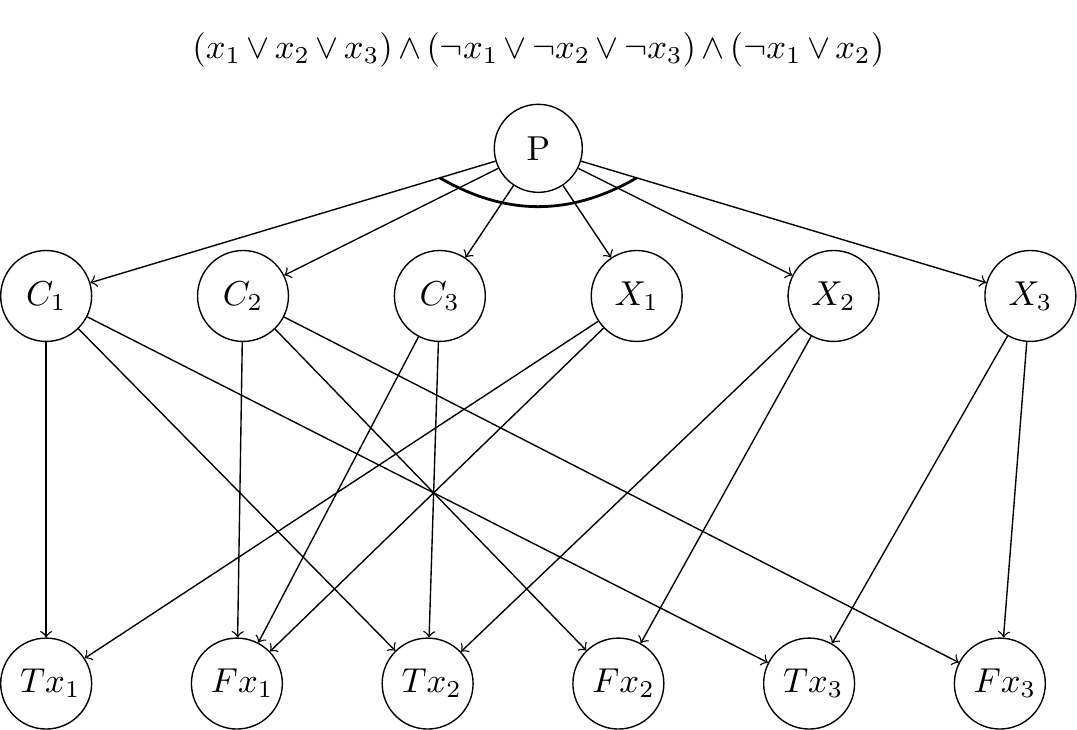}
	\includegraphics[scale=0.5]{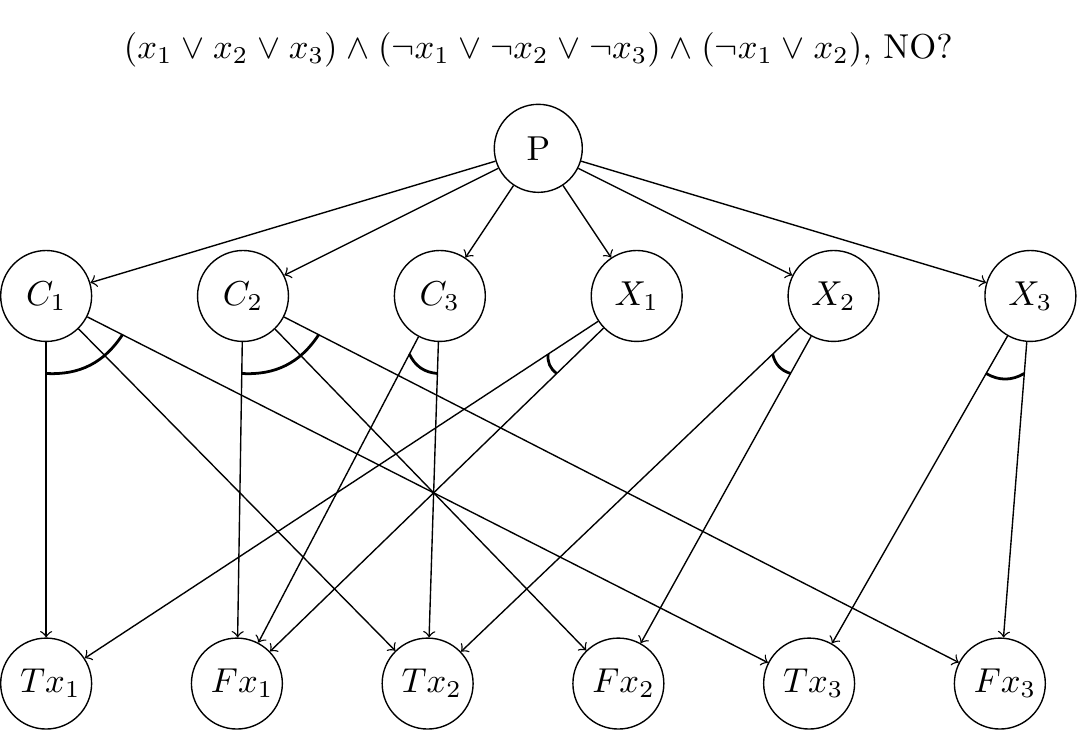}
	\caption{Above: Deciding whether this SAT instance is satisfiable can be reduced to finding the true proof number for $P$. Below: Deciding whether this SAT instance is unstatisfiable can be reduced to finding the true disproof number for $P$.}
	\label{fig:sat_to_andor_example}
\end{figure}

%


\section{Conclusions}
We proved that computing exact proof/disproof number for directed acyclic graphs is NP-hard. 
We expect our discussion could provide useful inspiration to future PNS developments, either for solving games or real-world AND/OR graphs.

\clearpage
\bibliographystyle{named}
\bibliography{ijcai21}

\end{document}